\documentclass{article}

\usepackage{amssymb,cite} 
\usepackage{comment}

\usepackage{multirow}
\usepackage{amssymb}
\usepackage{amsmath,amsthm,cite,enumerate}
\usepackage{wasysym}
\usepackage{enumitem}
\usepackage{tensor}

\newcommand{\mbold}[1]{\mbox{\boldmath{\ensuremath{#1}}}}   
\newcommand{\be}{\begin{equation}}
\newcommand{\ee}{\end{equation}}
\def \bea { \begin{eqnarray}}
\def \eea {\end{eqnarray}}
\def \R {\mathcal{R}}

\def \bl {\mbox{{\mbold\ell}}}
\def \bk {\mbox{{\mbold k}}}
\def \bn {\mbox{{\mbold n}}}

\newcommand{\M}[1]{{\stackrel{#1}{M}}}  
\newcommand{\Tb}{\mbold T}     
\newcommand{\m}[1]{{\stackrel{#1}{m}}}  

\def \b0 {\mbox{{\mbold 0}}}

\newcommand{\pul}{{\textstyle{\frac{1}{2}}}}

\newtheorem{theorem}{Theorem}[section]
\newtheorem{lemma}[theorem]{Lemma}
\newtheorem{propos}[theorem]{Proposition}
\newtheorem{corol}[theorem]{Corollary}

\theoremstyle{definition}

\newtheorem{remark}[theorem]{Remark}

\numberwithin{equation}{section}

\newcommand*\tho{\text{\thorn}}
\newcommand*\dho{\text{\dh}}
\DeclareMathOperator{\Tr}{Tr}

\def \Bell {\mbox{{\mbold\ell}}}
\def \Bm {\mbox{{\mbold m}}}
\def \Bn {\mbox{{\mbold n}}}

\newcommand*\Be{\ensuremath{\textbf{e}}}

\def \BI {\mbold I}
\def \BS {\mbold S}
\def \BA {\mbold A}
\def \Brho {\mbold \rho}
\def \Bomega {\mbold \Omega}
\def \Bomegap {\mbold \Omega'}

\newcommand*\Omegap{\ensuremath{\Omega^{\prime}}}
\newcommand*\omegap{\ensuremath{\omega^{\prime}}}
\newcommand*\omegat{\ensuremath{\tilde \omega}}

\begin{document}

\title{Spacetimes of Weyl and Ricci type N in higher dimensions}



\author{M. Kuchynka$^{\diamond,\star}$,  A. Pravdov\' a$^\star$\\
\vspace{0.05cm} \\
{\small $^\diamond$ Mathematical Institute of Charles University}, {\small Sokolovsk\' a 83, 186 75  Prague 8, Czech Republic}  \\
{\small $^\star$ Institute of Mathematics, Academy of Sciences of the Czech Republic}, \\ {\small \v Zitn\' a 25, 115 67 Prague 1, Czech Republic} \\
 {\small E-mail: \texttt{kuchynkm@gmail.com, pravdova@math.cas.cz}} }


\maketitle

\begin{abstract}
We study  the geometrical properties of null congruences generated by an aligned null direction of the Weyl tensor (WAND) in spacetimes of the Weyl and Ricci type N 
(possibly with a non-vanishing cosmological constant) in an arbitrary dimension.   
We prove that a  type N  Ricci tensor and a
 type III or N Weyl tensor have to be aligned. In such  spacetimes, the
multiple 
WAND has to be geodetic.
For spacetimes with  type N aligned Weyl and Ricci tensors, the canonical form of the optical matrix in the twisting and non-twisting case is derived
and the dependence of the Weyl  and the Ricci tensors and Ricci rotation coefficients on the affine parameter of the 
geodetic null congruence generated by the WAND is obtained. 
\end{abstract}

\section{Introduction and summary}

In general, the Einstein equations and their various generalizations in higher dimensions are notoriously difficult to solve. Therefore, when studying various classes of spacetimes one has to resort to certain assumptions which lead to a considerable simplification of the field equations. One possibility is to make simplifying assumptions on the curvature of the spacetime.  In particular, assuming an
algebraically special spacetime (i.e. a spacetime that admits a multiple Weyl aligned null direction (mWAND)) in the  Weyl alignment classification 
 \cite{Milsonetal05,Coleyetal04} (see also \cite{review} for a recent review) 
 leads to a substantial  simplification of the field equations. Moreover,  it is usually also assumed  that the spacetime is Einstein (the Ricci tensor is proportional to the metric). Under  these simplifying assumptions, various results of interest have been already obtained.
For example, a generalization of the (necessary part of) the Goldberg-Sachs theorem  was obtained in five dimension \cite{Ortaggioetal12}, which subsequently led to the determination of all the
algebraically special vacuum solutions of the Einstein equations in five dimensions 
(including new solutions) \cite{5DHarvey-notwist}, \cite{5DHarvey-Kerr} and \cite{5DHarvey-twist}.

In this work, we want to go beyond the Einstein spaces and we set out to study
  Ricci type N\footnote{In this paper, algebraic types are considered  to be genuine unless stated otherwise, thus, e.g. by `type N' we mean type N not including type O.} 
 spacetimes {(possibly with a non-vanishing cosmological constant 
added)\footnote{Thus, more precisely, in our case the traceless Ricci tensor is
of type N.}}, i.e. spacetimes with the Ricci tensor of the form 
\be \label{RicciIntro} 
 R_{ab} = \lambda g_{ab}+\eta {k_a k_b} ,
\ee
 where $k_a$ is a null vector, $\eta$ is a non-vanishing scalar 
({`radiation density'}) 
and $\lambda$ is proportional to the cosmological constant.\footnote{In fact,
$\lambda= 2\Lambda /(d-2)$, where $\Lambda$ is the cosmological constant. Note
that the `$\lambda$' term in the Ricci tensor \eqref{RicciIntro} 
does not affect the Bianchi equations.} 
 Four-dimensional solutions of the Einstein equations with the Ricci tensor of the form \eqref{RicciIntro} {with $\lambda =0$}, representing spacetimes    
containing pure radiation (null dust, e.g. null Maxwell field), have  already been studied in detail 
in the literature (see e.g. \cite{stephani} and 
\cite{griffiths}  and references therein). In the present paper, we  study the geometrical properties of  Weyl type  N (and partly also of  more general type III) spacetimes with the  Ricci tensor of the form  \eqref{RicciIntro} in higher dimensions. 

Our assumptions are geometrical and thus our results, apart from Einstein  gravity,  also hold in the context of generalized gravities,
such as e.g. quadratic gravity. However, note  that the matter content of spacetimes with the Ricci tensor \eqref{RicciIntro} may be theory-dependent. Indeed, a Weyl and Ricci type N metric
representing a pure radiation spacetime in Einstein gravity may represent a vacuum spacetime in quadratic gravity \cite{MalPra11}.  

For Weyl type N and III spacetimes, there is a unique preferred null direction - a multiple Weyl aligned null direction (mWAND) $\bl$. Since, in principle, the Ricci tensor \eqref{RicciIntro} defines another null direction $\bk$, there may be two distinct prefered null directions 
(i.e., the Weyl and the Ricci tensor may be non-aligned). However, 
in section \ref{align}, we show that {\em for Weyl type III and N spacetimes, the Ricci tensor (\ref{RicciIntro}) is necessarily aligned with the Weyl tensor}\footnote{{For brevity,  whenever the Weyl and the traceless Ricci tensors are (non)aligned, we say that
the spacetime is (non)aligned.}}
(i.e. $\bk \propto \Bell $,\footnote{Thus, for Weyl and Ricci type N spacetimes, there is always a frame where both the Weyl  and the  Ricci tensors admit only boost weight $(-2)$ components. This implies that all rank-2 tensors constructed from the Weyl and the Ricci tensor that are at least quadratic in the Weyl/Ricci tensor vanish. This implies that apart from the Einstein equations, these metrics also solve the field equations of various generalizations of the Einstein theory, such as Gauss-Bonnet or Lovelock gravity, with a type N  stress energy tensor ($T_{ab} \propto \ell_a \ell_b$), see e.g. \cite{turek1} and
	\cite{turek2} for higher-dimensional Weyl and Ricci type N spacetimes 
	as solutions of generalized gravity theories.} 
	see theorem \ref{nonaligned}), thus generalizing  an earlier result by Wils \cite{wils} which was obtained in four dimensions. 

We then proceed with studying the geodeticity of $\Bell$, arriving at the conclusion that {\em for a  Weyl type III and N spacetime with the Ricci tensor of the form \eqref{RicciIntro}, the common aligned null direction of the Weyl and the Ricci tensor is geodetic} (theorem \ref{geodeticANDs}).

A next natural step, which we take in sections \ref{sec_notwist}, \ref{sec_twist},  is to study geometrical properties of this preferred  geodetic null direction, which are encoded  
in the  optical matrix $\rho_{ij}$ (introduced in section \ref{sec_optic}). However, note  that here we limit ourselves to the type N case since the type III case has not yet been completely resolved  even in the simpler case of Einstein spacetimes 
(see \cite{bianchiclanek,bianchiclanekcor} and \cite{review} for partial results).

For Einstein type N spacetimes, the canonical form of the  optical matrix 
is
\cite{bianchiclanek,bianchiclanekcor,ghpclanek}
\begin{equation}\label{EinsteinN}
\boldsymbol{\rho} = s
\left(\begin{array}{cc|c}
1 & a & \multirow{2}{*}   {\mbold{0}} \\ 
-a & 1    \\ \hline
 \multicolumn{2}{c|}{\boldsymbol{0}} & \boldsymbol{0} \\
\end{array}
\right).
\end{equation}
In this case, the optical matrix obeys the so-called optical constraint $\rho_{ik} \rho_{jk} \propto \rho_{(ij)} $.   Note that for Einstein spacetimes, the optical constraint also holds in much more general situations (e.g. for Kerr-Schild metrics \cite{OrtPraPra09,MalPra11},  non-degenerate geodetic double WANDs in asymptotically
flat type II vacuum spacetimes \cite{OrtPraPra09b}, `general' non-twisting type II Einstein spacetimes \cite{clanek11}, etc).
This result was used in   \cite{OrtPraPra09} and  \cite{OrtPraPra10} to integrate certain Bianchi identities to determine the $r$-dependence (where $r$ is an affine parameter of the null congruence corresponding to the geodetic multiple WAND) 
of the curvature and to study asymptotic properties of such spacetimes.

We show that in contrast to the type N Einstein case,  the optical constraint does not hold if the 
the term $\eta k_a k_b$ is present in the  Ricci tensor 
(unless, trivially,  the spacetime is Kundt, i.e. $\boldsymbol{\rho=0}$). 
Nevertheless, the symmetric and skew-symmetric parts of the optical matrix $\rho_{ij}$ still commute and thus $\rho_{ij}$ is a normal\footnote{
Real normal matrix $A$ satisfies $A^T A=AA^T$.} matrix. Similarly to  the Einstein case, we find that the rank of $\rho_{ij}$ is at most 2, however, the $2\times 2$ block is now shearing  and thus the canonical form of $\rho_{ij}$ is  
(theorems \ref{twistfree} and \ref{twistingmatrix}) 
\begin{equation}\label{RicciNN}
\boldsymbol{\rho} = s
\left(\begin{array}{cc|c}
1 & a & \multirow{2}{*}   {$\boldsymbol{0}$} \\
-a & b    \\ \hline
 \multicolumn{2}{c|}{\boldsymbol{0}} & \boldsymbol{0} \\
\end{array}
\right).
\end{equation}
For the conformally flat case, the optical matrix vanishes completely (corollary \ref{c-flat}). 
Let us summarize possible subcases in  table \ref{matrixrho}.
\begin{table}[htb]
  \begin{center}
  \begin{tabular}{|c|l|l|} 
    \hline 
		rank of $\rho_{ij}$ &   conditions & spacetime \\ \hline 
		0 & $s=0$ & Kundt \\
1 & $s\not= 0$, $b=0$ & non-twisting\\
2 & $s\not= 0$, $b=1$ & Einstein \\
	2 & $s\not= 0$, $b\not= 0,1$, $a=0$ & non-twisting \\
	2 & $s\not= 0$, $b\not= 0,1$, $a\not=0$ & twisting \\ \hline
  \end{tabular}
  \caption{Possible subcases of the optical matrix \eqref{RicciNN} for spacetimes of the Weyl type N and 
	the Ricci tensor of the form \eqref{RicciIntro}. Note that the first ($s=0$) case
	also allows  the Weyl type O.}
  \label{matrixrho}
   \end{center}
\end{table}

For a  normal optical matrix $\rho_{ij}$, the Sachs equation (see \eqref{fullsachs}), 
which determines the evolution of the optical matrix along the multiple WAND $\bl$, can be integrated 
\cite{OrtPraPra10}.  Thus, in section \ref{alleq}, the $r$-dependence (where $r$ is the affine parameter
along $\bl$) of the  `radiation density' and
the Weyl tensor, as well as some Ricci rotation coefficients, is determined.

Finally, we briefly discuss a direct product of an algebraically special 
aligned Ricci type N  spacetime 
and a Euclidean space. For the resulting spacetime, we show that it is again an algebraically special spacetime of Ricci type N with the same WAND as that of the original spacetime. It turns out that the resulting spacetime does not change its primary Weyl type, unless it is conformally flat. We also show the connection between the geometry of the congruence generated by the multiple WAND in the original and the resulting spacetime. In particular, the resulting spacetime is necessarily shearing, unless the original spacetime is Kundt. 
The four-dimensional  pure radiation metrics given in table 3 can be used as seed metrics to construct higher-dimensional spacetimes of Weyl and Ricci type N with various optical properties (sec. \ref{priklad}).

\section{Preliminaries}
\label{sec_prelim}

\subsection{Null frames and Lorentz transformations}
Let $d$ denote the dimension of a spacetime. 
In a $d$-dimensional spacetime, we shall consider a \textit{null frame} $\{ \Be_{(a)} \}$, i.e. a real frame 
\begin{equation}\label{frame}
\{  \Bell \equiv \Be_{(0)} = \Be^{(1)} ,\ 
\Bn \equiv \Be_{(1)} = \Be^{(0)} ,\ \Bm_{(i)} \equiv \Be_{(i)} = \Be^{(i)}  \}
\end{equation}
with two null vector fields $\Bell$ and $\Bn$ and $d-2$ spacelike vector fields 
$\Bm_{(i)}$ 
such that they satisfy the following relations 
\begin{equation}\label{ortogonality}
\ell_a \ell^a = n_a n^a = \ell_a m_{(i)}^a = n_a m_{(i)}^a = 0, \qquad
\ell_a  n^a = 1, \qquad 
 m_a^{(i)}  m_{(j)}^a = \delta_j^i .
\end{equation}
The metric expressed in terms of dual null frame vectors thus reads 
\begin{equation}\label{metric}
g_{ab} = 2\ell_{(a} n_{b)} + \delta_{ij} m_a^{(i)} m_b^{(j)}.
\end{equation}
Throughout the paper, we make use of two types of indices: indices $a,b,\dots$ take values $0,\dots,d-1$, while indices $i,j,\dots$ take values $2,\dots,d-1$, unless  stated otherwise. We employ the Einstein summation convention for both types of indices, however, since the frame indices $i,j,\dots$ are raised and lowered  by $\delta_{ij}$, we do not distinguish between the covariant and contravariant null frame tensor components  corresponding to indices $i,j,\dots$. When we want to emphasize that there is no summation over repeated indices, we put them both in brackets.

Relations \eqref{ortogonality} are preserved under local Lorentz transformations acting on a tangent space. Every proper orthochronous Lorentz transformation can be decomposed into 
\textit{boost} with a positive function $\lambda$
\begin{equation}\label{boost}
\Bell \mapsto \lambda \Bell , \qquad \Bn \mapsto \lambda^{-1} \Bn, \qquad 
\Bm_{(i)} \mapsto  \Bm_{(i)},
\end{equation}
\textit{spin} determined by $X_{ij} \in SO(d-2)$ at a given point of the spacetime
\begin{equation}\label{spin}
\Bell \mapsto \Bell, \qquad 
\Bn \mapsto \Bn, \qquad
 \Bm_{(i)} \mapsto X_{ij} \Bm_{(j)},
\end{equation}
and \textit{null rotations} about vectors $\Bell$ or $\Bn$ determined by a set of (real) functions $z_i$
\begin{equation}\label{nullrotation}
\begin{aligned}
\Bell \mapsto \Bell, \qquad 
\Bn \mapsto \Bn + z^i \Bm_{(i)} - \frac{1}{2} z^i z_i \Bell, \qquad
\Bm_{(i)} \mapsto \Bm_{(i)} - z_i \Bell,\\
\Bn \mapsto \Bn, \qquad 
\Bell \mapsto \Bell + z^i \Bm_{(i)} - \frac{1}{2} z^i z_i \Bn, \qquad
\Bm_{(i)} \mapsto \Bm_{(i)} - z_i \Bn,
\end{aligned}
\end{equation}
respectively. 

\subsection{Optical matrix and optical scalars}
\label{sec_optic}
Consider a covariant derivative $L_{ab} \equiv \nabla_b \ell_a$ of the null frame vector $\Bell$. Projecting $L_{ab}$ on the frame vectors  $\Be_{(a)}$, 
one obtains its null frame components 
\begin{equation}\label{Lcomponents}
L_{(a)(b)} \equiv L_{cd} e_{(a)}^c e_{(b)}^d.
\end{equation}
The vector field $\Bell$ is tangent to a geodetic null congruence if and only if 
$\kappa_i\equiv L_{(i)(0)} = 0$. Then, such geodetic congruence can be always affinely parametrized, i.e. $L_{(1)(0)}=0$.

 Similarly, one can project covariant derivatives of the remaining frame vectors,
\begin{equation}
N_{ab} = \nabla_b n_a, \qquad \M{(i)}_{ab} = \nabla_b m^{(i)}_{a},
\end{equation}
 into the basis \eqref{frame} to obtain the scalars  
$N_{(a)(b)}$, $\M{(i)}_{(a)(b)}$. From  (\ref{ortogonality}), 
these scalars satisfy the identities
\begin{equation}\label{eqn:ident1}
  \quad\quad N_{(0)(a)} + L_{(1)(a)} = 0, \quad \M{(i)}_{(0)(a)} + L_{(i)(a)} = 0,
  \quad \M{(i)}_{(1)(a)} + N_{(i)(a)} = 0, \quad \M{(i)}_{(j)(a)} + \M{(j)}_{(i)(a)} = 0,
\end{equation}
and
\begin{equation}\label{eqn:ident2}
  L_{(0)(a)} = N_{(1)(a)} = \M{(i)}_{(i)(a)} = 0.
\end{equation}

For a geodetic null field $\bl$,  the so called optical matrix  $\rho_{ij} \equiv L_{(i)(j)}$  contains information about certain geometric properties of the null congruence, namely, one can define    the \textit{optical scalars} $\sigma$ (\textit{shear}), $\theta$ (\textit{expansion}) and $A$ (\textit{twist}) of $\Brho$ in the following way: 
\begin{equation}
\sigma^2 \equiv \sigma_{ij} \sigma^{ij}, \qquad 
\theta \equiv \frac{1}{d-2} \rho \indices{^i_i},\qquad
A^2 \equiv A_{ij}A^{ij},
\end{equation}
where  $\sigma_{ij} \equiv S_{ij}-\theta \delta_{ij}$  is the traceless part of the symmetric  part $\BS$ of $\Brho$ and $\BA$ is the  skew-symmetric part of $\Brho$. 
We say that a spacetime is non-twisting or twisting if $A_{ij} =0$
or $A_{ij} \neq 0$, respectively.

\subsection{GHP formalism}
Throughout the paper, we use some of the features of the higher-dimensional Geroch–Held–Penrose (GHP) formalism that was developed in \cite{ghpclanek}. 
The basic notion of GHP formalism is a \textit{GHP scalar}. We say that an object 
$T_{i_1 \dots i_s}$ is a GHP scalar of spin weight $s$ and boost weight $b$ if it transforms as 
\begin{equation}
T_{i_1 \dots i_s} \mapsto X_{i_1 j_1} \dots X_{i_s j_s} T_{j_1 \dots j_s}
\end{equation}
under spins $X_{ij}$ and as
\begin{equation}
T_{i_1 \dots i_s} \mapsto \lambda^b T_{i_1 \dots i_s}
\end{equation}
under boosts $\lambda$. Note that, for example, the frame components 
\bea
&& \omega'\equiv R_{(1)(1)} \equiv R_{ab}n^a n^b,\ \ \label{omegap}\\
&& \omega\equiv R_{(0)(0)} \equiv R_{ab}\ell^a \ell^b,\label{omega}
\eea
 of the Ricci tensor $R_{ab}$ are GHP scalars of spin and boost weights $0$, $(-2)$ and
$0$, $(+2)$, respectively, while the matrix  
\be
\Omega'_{ij}\equiv C_{(1)(i)(1)(j)} \equiv C_{abcd}n^a m_{(i)}^b n^c m_{(j)}^d
\label{Omega}
\ee
 of frame components of the Weyl tensor is a GHP scalar of spin weight $2$ and boost weight $(-2)$.

For GHP scalars,\footnote{For brevity, we omit brackets in the scalar indices.} we introduce  the notation following \cite{ghpclanek}, shown in table \ref{tab:weights}.
\begin{table}[ht]
 \begin{center}
   \begin{tabular}{|c|c|c|c|l|}
    \hline Quantity & Notation & Boost weight $b$ & Spin $s$ & Interpretation\\ [1mm]\hline
    $L_{ij}$  & $\rho_{ij}$  & 1  & 2 & expansion, shear and twist of $\bl$\\[1mm]
     $L_{ii}$  & $\rho=\rho_{ii}$  & 1  & 0 & expansion of $\bl$\\[1mm]
    $L_{i0}$  & $\kappa_{i}$   & 2  & 1 & non-geodesity of $\bl$\\[1mm]
    $L_{i1}$  & $\tau_{i}$   & 0  & 1 & transport of $\bl$ along $\bn$\\[1mm]
    $N_{ij}$  & $\rho'_{ij}$ & -1 & 2 & expansion, shear and twist of $\bn$\\[1mm]
     $N_{ii}$  & $\rho'=\rho'_{ii}$ & -1 & 0 & expansion of $\bn$\\[1mm]
     $N_{i1}$  & $\kappa'_{i}$  & -2 & 1 & non-geodesity of $\bn$\\[1mm]
    $N_{i0}$  & $\tau'_{i}$  & 0  & 1 & transport of $\bn$ along $\bl$\\[1mm]\hline
  \end{tabular}
  \caption{\label{tab:weights}A higher-dimensional generalization of GHP scalars corresponding to the Ricci rotation coefficients \cite{ghpclanek}.}
 \end{center}
\end{table}

{In type III spacetimes, there is a frame such that the type III Weyl tensor has the 
form\footnote{The operation $\{ \}$ is defined as $2 T_{\{ abcd\} } \equiv T_{[ab][cd]} + T_{[cd][ab]}$. }

\begin{equation}
C_{abcd} = 8 \Psi_{i}^\prime \ell_{\{a} n_{b} \ell_{c} m_{d\}}^{(i)} + 
		4 {\Psi}_{ijk}^\prime  \ell_{\{a} m_b^{(i)} m_c^{(j)} m_{d\}}^{(k)} 
+ 4 \Omegap_{ij} \ell_{\{a} m_{b}^{(i)} \ell_c m_{d\}}^{(j)}.\label{WeylIII}
\end{equation}
where ${\Psi}_{ijk}^\prime=-{\Psi}_{ikj}^\prime$, ${\Psi}_{i}^\prime={\Psi}_{jij}^\prime$ and ${\Psi}_{[ijk]}^\prime=0$.
For the type N Weyl tensor, only terms with $\Omegap_{ij}$ are non-vanishing, i.e. 
\begin{equation}\label{weyltensor}
C_{abcd} = \Omegap_{ij}  \ell_{\{a} m_b^{(i)} \ell_c m_{d\}}^{(j)}, 
\end{equation}
with $\Omegap_{ij}$ being symmetric and traceless.}

Further, one can define a \textit{GHP derivative operators} $\tho, \tho^\prime$ and $\dho_i$ acting on  a GHP scalar $\Tb$ of boost weight $b$ and spin $s$  and creating again a GHP scalar
(see (2.15)--(2.17) in \cite{ghpclanek}):
  \begin{eqnarray}
    \tho T_{i_1 i_2...i_s} &\equiv & D T_{i_1 i_2...i_s} - b L_{(1)(0)} T_{i_1 i_2...i_s} 
                                     + \sum_{r=1}^s \M{(k)}_{(i_r) (0)} T_{i_1...i_{r-1} k i_{r+1}...i_s},
																		\label{thorn}\\
    \tho' T_{i_1 i_2...i_s} &\equiv & \Delta T_{i_1 i_2...i_s} - b L_{(1)(1)} T_{i_1 i_2...i_s} 
                                     + \sum_{r=1}^s \M{(k)}_{(i_r)(1)} T_{i_1...i_{r-1} k i_{r+1}...i_s},\\
    \eth_i T_{j_1 j_2...j_s} &\equiv & \delta_i T_{j_1 j_2...j_s} - b L_{(1)(i)} T_{j_1 j_2...j_s} 
                                     + \sum_{r=1}^s \M{(k)}_{(j_r) (i)} T_{j_1...j_{r-1} k j_{r+1}...j_s},\label{eth}
  \end{eqnarray}
	where the derivatives along the frame vectors are defined as
\be
D \equiv \ell^a \nabla_a, \qquad \bigtriangleup  
\equiv n^a \nabla_a, \qquad \delta_i \equiv m^{a}_{(i)} \nabla_a .\label{deriv}
\ee

For more details and for other features of the higher-dimensional GHP formalism, see \cite{ghpclanek}. 

In this paper, projections of the Bianchi identity 
\begin{equation}
\nabla_{[a}R_{bc]de}=0
\end{equation}
onto the null frame \eqref{frame} (\!\!\cite{ghpclanek}, see also \cite{bianchiclanek}) are used.
The complete set of all the independent null frame components of the Bianchi identity for algebraically special Einstein spacetimes using the GHP notation was given 
in \cite{ghpclanek}.  The Bianchi equations with the non-trivial Ricci tensor are obtained by simple replacement (2.43)--(2.48) in \cite{ghpclanek}.

Now, one has to distinguish between an {\it aligned} and a {\it non-aligned} case. 
In the {aligned} case, one has $\langle \Bell \rangle = \langle \bk 
\rangle$\footnote{Here, $\langle \Bell \rangle$ denotes the equivalence class of vectors having the same direction as $\Bell $, i.e. $\langle \Bell \rangle = \langle \bk 
\rangle$ iff $\Bell \propto \bk 
$.}, i.e. 
the Weyl and the  
Ricci tensors have a common aligned null direction. Otherwise, we say that 
the Weyl and the   
Ricci tensors are non-aligned. 
If this is the case, one can further choose the null frame vector $\bn$ such that $\bn \propto \bk $.

For  the {\it non-aligned} 
Ricci tensor \eqref{RicciIntro}, where $\bn=\bk$,
 the only necessary replacement in the Bianchi equations   is 
 (2.43) in \cite{ghpclanek},  i.e.
\begin{equation}\label{substitut}
\Omega_{ij} \to \Omega_{ij} + \omegat \delta_{ij}.
\end{equation}
 Here, $\omegat$ is defined as 
\be
\omegat \equiv \frac{\eta}{d-2}.
\ee
The Bianchi equations (B2), (B3) and (B4) in \cite{ghpclanek} then read
\begin{equation}\label{primedB2}
\tho^\prime \omegat \delta_{ij} = - (\Psi_{j}^\prime \delta_{ik} - \Psi_{jik}^\prime) \kappa_k 
- \omegat \delta_{ik} \rho_{kj}^\prime,
\end{equation}
\begin{equation}\label{primedB3}
  \Psi_{[i|kl|}^\prime \kappa_{j]} +  \Psi_{[k|ij|}^\prime \kappa_{l]} + \omegat \delta_{i[k} \rho_{|j|l]}^\prime 
-\omegat \delta_{j[k} \rho_{|i|l]}^\prime =  0,
\end{equation}
\begin{equation}\label{primedB4}
\omegat \delta_{i[j} \rho_{kl]}^\prime = 0,
\end{equation}
respectively. Apart from the Bianchi equations, we will also employ equations (2.50) and (2.51)  of \cite{ghpclanek}, {which for the Ricci tensor of the form \eqref{RicciIntro} reduce to} 
\bea\label{2.50}
 \tho^\prime \omega &=& - \rho^\prime \omega,\\
\kappa'_i \omega &=&0,\label{geodk}
\eea
where $\rho^\prime$ denotes the trace of $\rho_{ij}^\prime$ and
 $\omega=\eta$.

For the {\it aligned} Ricci  \eqref{RicciIntro} and Weyl  \eqref{weyltensor} tensors, where  $\bl=\bk$,
 the only necessary replacement in the Bianchi equations   is (2.43')
 in \cite{ghpclanek}, i.e.
\begin{equation}\label{substitutp}
\Omegap_{ij} \to \Omegap_{ij} + \omegat \delta_{ij}.
\end{equation}
We will employ equations (B2'), (B3') and (B4') of \cite{ghpclanek} 
\begin{equation}\label{B1}
\tho \Omegap_{ij} + \tho \omegat \delta_{ij} = - \Omegap_{ik}\rho_{kj} - \omegat\rho_{ij},
\end{equation}
\begin{equation}\label{B2}
\Omegap_{i[k}\rho_{|j|l]} + \omegat \delta_{i[k} \rho_{|j|l]} = 
\Omegap_{j[k}\rho_{|i|l]} + \omegat \delta_{j[k} \rho_{|i|l]},
\end{equation}
\begin{equation}\label{B3}
 \Omegap_{i[j}\rho_{kl]} + \omegat \delta_{i[j} \rho_{kl]} = 0
\end{equation}
and  (2.50') and (2.51) in \cite{ghpclanek} 
\bea
 \tho \omega' &=& - \rho \omegap,\label{2.50p}\\
\kappa_i \omegap &=&0,\label{2.51p}
\eea
where $\omegap=\eta$.


\section{ 
Alignment and geodeticity for Weyl type III and N,  Ricci type N spacetimes}
\label{align}

In \cite{wils}, it is proven that non-aligned pure radiation spacetimes of the Weyl type  III and N do not exist in four dimensions (see \cite{wils}, Theorem 3). It turns out that this is also true in arbitrary dimension.

\begin{propos}\label{nonaligned}
Weyl type III and N spacetimes with the Ricci tensor of the form \eqref{RicciIntro}  are necessarily aligned.
\end{propos}
\begin{proof} 
Thanks to Theorem 3 of \cite{wils}, it is sufficient to prove the assertion for $d>4$. 
Let $\Bell$ be a multiple WAND of a spacetime.  
Suppose that $\bk $ is the aligned null direction of the Ricci tensor such that 
$\langle \bk \rangle \neq \langle \Bell \rangle$. 
Then, one can consider a null frame $\{ \Be_{(a)} \}$ with $\Be_{(0)} \equiv \Bell$ and $\Be_{(1)} \equiv \bk $.  
In this null frame, the Ricci tensor has the form \eqref{RicciIntro} (with non-vanishing $\eta$), while the Weyl tensor takes  the form \eqref{WeylIII}.

First, by contracting equation \eqref{primedB4} with respect to $i$ and $j$, one obtains 
\begin{equation}
(d-4) \rho_{[kl]}^\prime = 0,
\end{equation}
i.e. $\bk $ is non-twisting for every $d>4$ 
(and thus also geodetic).\footnote{For a non-twisting null congruence $\bk$, one can introduce a foliation $u=$const.,
such that $k_a=u,_a$. Then $k_{a;b}k^b=u_{;ab}k^b=u_{;ba}k^b=k_{b;a}k^b=0$. 
In the non-aligned case,  
this can be also seen directly from \eqref{geodk}.}
Now, we will prove that  the optical matrix $\rho_{ij}^\prime$
associated with $\bk$ in fact vanishes.
Substituting \eqref{2.50} into \eqref{primedB2}, one has that 
\begin{equation}\label{B2plus2.50}
-\rho^\prime \omegat \delta_{ij} = - \Psi_{j}^\prime \kappa_{i} + \Psi_{jik}^\prime \kappa^k - \omegat \rho_{ij}^\prime.
\end{equation}
Then, taking the trace of \eqref{B2plus2.50}, one arrives at
\begin{equation}\label{rce2}
2 \Psi_{i}^\prime \kappa^i = (d-3) \rho^\prime \omegat,
\end{equation}
while the contraction of \eqref{B2plus2.50} with $\kappa^i \kappa^j$ gives 
\begin{equation}\label{rce4}
\rho^\prime \omegat \kappa_i \kappa^i = \Psi_{j}^\prime \kappa^j \kappa_i \kappa^i 
- \Psi_{jik}^\prime \kappa^i \kappa^j \kappa^k + \omegat \rho_{ij}^\prime \kappa^i \kappa^j.
\end{equation}
Using \eqref{rce2} and $\Psi_{ijk}^\prime = - \Psi_{ikj}^\prime$, equation \eqref{rce4}  simplifies to 
\begin{equation}\label{rce3}
2 \rho_{ij}^\prime \kappa^i \kappa^j = -(d-5) \rho^\prime \kappa_i \kappa^i.
\end{equation}
Tracing equation \eqref{primedB3} in $i$ and $k$ leads to 
\begin{equation}\label{rceskoro5}
2 \Psi_{(l}^\prime \kappa_{j)} + 2 \Psi_{(j|k|l)}^\prime \kappa^k - (d-4) \omegat \rho_{jl}^\prime 
- \omegat \rho^\prime \delta_{jl} = 0.
\end{equation}
Contracting \eqref{rceskoro5} with $\kappa^j \kappa^l$, one obtains
\begin{equation}\label{rce5}
2 \Psi_{l}^\prime \kappa^l \kappa_j \kappa^j - (d-4) \omegat \rho_{jl}^\prime \kappa^j \kappa^l 
- \omegat \rho^\prime \kappa_j \kappa^j = 0.
\end{equation}
Finally, substituting  \eqref{rce2} and \eqref{rce3} into \eqref{rce5}, one obtains
\begin{equation}\label{rce6}
(d-3)(d-4) \omegat \rho^\prime \kappa_i \kappa^i = 0,
\end{equation}
thus, either $\rho^\prime = 0$, i.e. $\bk $ is non-expanding, and Proposition 1 in \cite{ricciclanek} implies $\rho_{ij}^\prime = 0$, 
or $\kappa_i = 0$, i.e. $\Bell$ is geodetic, and equation \eqref{rceskoro5} reduces to 
\begin{equation}\label{rceskoro5reduced}
 (d-4) \omegat \rho_{jl}^\prime 
+ \omegat \rho^\prime \delta_{jl} = 0.
\end{equation}
Its  trace implies  $\rho^\prime = 0$.
Thus, we again arrive at $\rho_{ij}^\prime = 0$.
Since assumptions of Proposition 2 of \cite{ricciclanek} are met, 
$\bk $ 
must be a multiple WAND. However, since we assume that the spacetime is of Weyl type III or N with (the only) multiple WAND $\Bell$, 
one has that necessarily $\langle \bk \rangle = \langle \Bell \rangle$, which is a contradiction.
\end{proof}

Thus,  the type III or N Weyl tensor and the Ricci tensor \eqref{RicciIntro}  
are always aligned,
and, without loss of generality, we can set  $\bk = \Bell$.
The boost weight $(-2)$ component \eqref{omegap} of the  Ricci tensor \eqref{RicciIntro} is then
$\omegap=\eta$.

Let us study geodeticity of the aligned direction $\bl$.
Equation \eqref{2.51p}   immediately implies $\kappa_i =0$ and thus  $\bl$ is  geodetic.
Therefore
\begin{propos}\label{geodeticANDs}
In a  Weyl type III or N 
spacetime with the Ricci tensor of the form \eqref{RicciIntro}, the common aligned null direction of the Weyl and the  Ricci tensors is geodetic. 
\end{propos}

In the rest of the paper, we focus solely on Weyl type N 
spacetimes with the Ricci tensor of the form \eqref{RicciIntro}.

Recall that $\BS$ and $\BA$ denote the symmetric and skew-symmetric parts of $\Brho$, respectively. Before we proceed any further, we give the following characterization 
of the non-twisting spacetimes of the  Weyl N with the Ricci tensor of the form \eqref{RicciIntro}.

\begin{lemma}\label{lemma1} 
A  spacetime of the Weyl type N  
with the Ricci tensor of the form \eqref{RicciIntro} 
 is non-twisting if and only if its $\Bomegap$ and $\BS$ commute.
\end{lemma}
\begin{proof}
By taking the trace of both sides of  equation \eqref{B3} in $i$ and $k$, we obtain the matrix relation
\begin{equation}\label{4}
\Bomegap \BA + \BA \Bomegap = (d-4) \omegat \BA.
\end{equation}
The skew-symmetric part of  equation \eqref{B1} reads 
\begin{equation}
0 = [\Bomegap , \BS ] + \Bomegap \BA + \BA \Bomegap + 2\omegat \BA.
\end{equation}
Using \eqref{4}, this equation further simplifies to 
\begin{equation}\label{7}
\omegap \BA = - [\Bomegap, \BS ].
\end{equation}
Since $\omegap$ is non-vanishing, this proves our assertion.
\end{proof}

We see that, in contrast to the Einstein spacetimes of Weyl type N, where $\Bomegap$ and $\BS$ always commute (see \cite{ghpclanek}), in spacetimes of  Weyl  
N with the Ricci tensor of the form \eqref{RicciIntro}, $\Bomegap$ and $\BS$ can be simultaneously diagonalized by Lorentz spins only in the non-twisting case. Thus, it is convenient to treat the twisting and the non-twisting case separately.

\section{Twist-free spacetimes}
\label{sec_notwist}
In this section, we consider spacetimes with a non-twisting multiple WAND. In \cite{clanek11}, the canonical form of the optical matrix corresponding to a non-twisting WAND was obtained for algebraically special Einstein spacetimes of dimension $d\geq5$
(for results on five dimensions, see also \cite{Ortaggioetal12} and for type III and N, see \cite{bianchiclanek}). It turned out that for these spacetimes, the optical matrix has at least one double eigenvalue, unless they are conformally flat. However, we observe that for non-twisting spacetimes of Weyl 
type N and with the Ricci tensor of the form \eqref{RicciIntro}, this is not the case. 

\begin{propos} \label{twistfree}
Let $\Bell$ be a multiple WAND of a twist-free spacetime of the Weyl  
type N with the Ricci tensor of the form \eqref{RicciIntro}. Then the canonical form of the corresponding optical matrix is 
\begin{equation}\label{canon-twistfree}
\boldsymbol{\rho} = s
\left(
\begin{array}{cc|c}
1 & 0 & \multirow{2}{*}{$\boldsymbol{0}$} \\
0 & b    \\ \hline
 \multicolumn{2}{c|}{\boldsymbol{0}} & \boldsymbol{0} \\
\end{array}
\right).
\end{equation}
 where $ b\not=1$, otherwise the spacetime is Einstein,\footnote{
It follows from
\eqref{S3} that $b=1$ iff $\omegat=0=R_{11}$ and thus an Einstein 
spacetime corresponds to $b=1$.} for $s=0$, it is Kundt.
\end{propos}
Note that $\boldsymbol{\rho}$ \eqref{canon-twistfree} is a normal matrix, not satisfying 
the optical constraint.

\begin{proof}
Let $d > 4$. 
From lemma \ref{lemma1} we know that $\Bomegap$ and $\BS$ commute. Thus, the null frame 
$\{ \Be_{(a)} \}$ can be transformed by Lorentz spin \eqref{spin} to a null frame $\{ \hat \Be_{(a)}  \}$ such that the matrices 
$\hat{\Bomega}^{\prime}$, $\hat \BS$  of the null frame components in  $\{ \hat \Be_{(a)}  \}$
are simultaneously diagonal. For the sake of clarity, let us denote $\hat{\Bomega}^{\prime} $ and $\hat \BS$ again by $\Bomegap$ and $\BS$, respectively. 
Thus $\Bomegap$ and $\BS$ take the form $\Bomegap = \text{diag}(\Omegap_{2},\dots,\Omegap_{d-1})$ and 
$\BS = \text{diag}(S_2,\dots,S_{d-1})$.
 Taking the trace of both sides of \eqref{B2} in $i$ and $k$, we obtain 
\begin{equation}\label{5}
\rho 
\left(\Bomegap + \omegat\BI\right) = \Bomegap \Brho + \Brho \Bomegap - (d-4) \omegat \Brho.
\end{equation}
Decomposing $\Brho$ into its symmetric and skew-symmetric part, $\Brho = \BS + \BA$, and using \eqref{4} one obtains 
\begin{equation}\label{6}
\rho \left(\Bomegap + \omegat\BI \right) = \Bomegap \BS + \BS \Bomegap - (d-4) \omegat \BS.
\end{equation}
Hence from \eqref{6} we obtain  the following relation between $S_i$ and $\Omegap_i$:
\begin{equation}\label{6*}
S_{i} \left (2 \Omegap_i - (d-4)\omegat \right ) = \rho  
(\Omegap_i + \omegat).
\end{equation}
Firstly, we will prove that $2 \Omegap_{i} \neq (d-4)\omegat$ for all $i$, unless $\Brho = \boldsymbol{0}$. 
Now, assume that there exists $i$ such that $2 \Omegap_{i} = (d-4)\omegat$. Then, using 
 \eqref{2.50p},  
we obtain
\begin{equation}
\tho \Omegap_i = -\frac{d-4}{2} \omegat \rho. 
\end{equation}
However, equation \eqref{6*} implies that $ \rho=0 $ 
and thus both $\tho \Omegap_i$ and $\tho \omegat$ vanish. 
Hence the Bianchi equation \eqref{B1} reduces to (no summation)
\begin{equation}
\Omegap_{(i)} S_{(i)} + \omegat S_i = 0,
\end{equation}
which, using the relation for $\Omegap_i$, immediately implies that $S_i = 0$ for such $i$. 
Now, considering equation \eqref{B2} for the choice of indices $i = k$, $i \neq j$, $i \neq l$ and $j = l$, one has 
\begin{equation}\label{11}
(\Omegap_{k}+ \omegat) S_{j} = - (\Omegap_j + \omegat) S_k  
\end{equation} 
for all $j \neq i$. Thus for $k=i$, one obtains that $(\Omegap_i + \omegat) S_j = 0$ for all $j\neq i$. This implies that 
$S_j = 0$ for all $j\neq i$ and consequently $\Brho = \boldsymbol{0}$, i.e. the spacetime belongs to the Kundt class.

Let us discuss the case $2 \Omegap_{i} \neq (d-4)\omegat$ for all $i$. In this case, it is possible to express $S_{i}$ using relation \eqref{6*} as 
\begin{equation}
S_{i}  = \rho 
\frac{\Omegap_i + \omegat}{2\Omegap_i - (d-4)\omegat}.\label{Si}
\end{equation}
Now, one can distinguish between four possible scenarios which lead to three 
qualitatively different canonical forms of $\Brho$. Namely, these are:  
\begin{enumerate}[label=(\roman*)]
\item\label{jednax} $\Omegap_i = -\omegat$ for all $i$;
\item \label{dvax} There is exactly one $i$ such that $\Omegap_i \neq -\omegat$;
\item \label{trix} There are exactly two $i$ such that $\Omegap_i \neq - \omegat$;
\item \label{ctyrix} There are at least three $i$ such that $\Omegap_i \neq -\omegat$;
\end{enumerate} 
One immediately obtains that  scenario \ref{jednax} leads to $\Omegap_i =0 =\omegat$
thanks to $\Omega'_{ii}=0$ and this contradicts our assumption of non-vanishing
$\omegat$.

Scenario \ref{dvax} 
leads to one possible non-vanishing component $S_i$ of $\BS$ and since $S_j = 0$ for all $j \neq i$, we obtain that 
necessarily $S_i = \rho$  
(with $\Omega'_i=(d-3)\omegat$). 

The third scenario leads to the possibility of two non-vanishing components, say $S_2$ and $S_3$, of $\BS$. In addition, one can use expression \eqref{Si} to verify that $S_2 \neq S_3$, otherwise 
$\Brho = \boldsymbol{0}$. In fact, from \eqref{Si} (also using $\Omega'_{ii}=0$) one can  determine that 
\be
S_3=S_2\frac{\Omega'_2-(d-3)\omegat}{\Omega'_2+\omegat}.\label{S3}
\ee

It remains to show that scenario \ref{ctyrix} 
leads to the vanishing optical matrix $\Brho$. 
Denote $I$ the set of all indices $i$ such that $\Omegap_i \neq - \omegat$. 
Then, using \eqref{Si}, equation \eqref{11} implies that 
\begin{equation}
\rho 
\left[ 2\Omegap_i -(d-4)\omegat \right] = -\rho 
\left[ 2\Omegap_j -(d-4)\omegat \right] 
\end{equation}
for all $i,j \in I,\ i\neq j$. Since $\rho=0$  
leads to $\Brho = \boldsymbol{0}$, assume that $\rho $  
is non-vanishing. 
Then 
\begin{equation}
\Omegap_i + \Omegap_j = (d-4)\omegat
\end{equation}
for all $i,j \in I,\ i \neq j$. This immediately implies that $2\Omegap_i = (d-4)\omegat$ for all $i \in I$, which contradicts our assumption. Thus $\rho=0$  
and consequently $\Brho = \boldsymbol{0}$. This completes the proof for the case $d>4$. The proof of the case $d=4$ is similar and one obtains that $S_3=S_2 (\Omega'_2 -\omegat)/(\Omega'_2 +\omegat)$ or $\rho_{ij}=0$.
\end{proof}

Also note that for spacetimes with $\Bomegap=\boldsymbol{0}$ and non-vanishing $\omegat$, 
equation \eqref{6} implies that $\rho=0$  
and Bianchi equation \eqref{B1} then gives  $\Brho = \boldsymbol{0}$. 
Thus we conclude with the following remark.
\begin{remark}\label{cflat}
A  non-twisting conformally flat spacetime  
with the Ricci tensor of the form \eqref{RicciIntro} is necessarily a Kundt spacetime 
{(see  \cite{HDVSI} for a discussion of such spacetimes).} 
\end{remark}

\section{Twisting spacetimes}
\label{sec_twist}
In this section, we consider spacetimes with a twisting multiple WAND. 
In \cite{ricciclanek}, it was proven that in spacetimes of odd dimensions, a twisting geodetic WAND must be also shearing. 
We show that for spacetimes of  
Weyl type N with the Ricci tensor of the form \eqref{RicciIntro}, this result extends to  
{even} dimensions. 
\begin{propos} 
Twisting, shear-free spacetimes of Weyl 
 type N with the Ricci tensor of the form \eqref{RicciIntro} do not exist. 
\end{propos}
\begin{proof}
Let $d>4$ and let $\{ \hat \Be_{(a)}\}$ be the null frame, in which $\BS$ is diagonal, and let us 
employ the notation used in the proof of proposition \ref{twistfree}. 
Consider  equations \eqref{B2} and \eqref{B3} for the choice of indices $i=k$,  $i\neq j$, $i \neq l$ and $j\neq l$ to obtain 
\begin{equation}\label{12}
\left( \Omegap_i + \omegat \right) A_{jl} + \Omegap_{l(i)} A_{(i)j} - \Omegap_{j(i)} A_{(i)l} = - \Omegap_{jl} S_{i}
\end {equation}
and 
\begin{equation}\label{13}
\Omegap_{(i)j} A_{(i)l} + \Omegap_{(i)l} A_{j(i)} - \left( \Omegap_i + \omegat \right) A_{jl} = 0.
\end {equation}
Combining \eqref{12} and \eqref{13}, one immediately obtains that
\begin{equation}\label{14}
S_{i} \Omegap_{jl} = 0
\end{equation}
for all $i \neq j$, $i \neq l$, $j \neq l$.
Rewriting the equation \eqref{6} as 
\begin{equation}\label{6**}
\rho 
\left(\Bomegap + \omegat\BI \right) = [\Bomegap, \BS ] + 2 \BS \Bomegap - (d-4) \omegat \BS
\end{equation}
and using \eqref{7}, we obtain 
\begin{equation}\label{15}
\rho 
\left(\Bomegap + \omegat \BI \right) =  2 \BS \Bomegap - \omegap \BA- (d-4) \omegat \BS.
\end{equation}
Then, considering \eqref{15} for the choice of the indices $i \neq j$, one obtains the relation for the components of $\BA$:
\begin{equation}\label{16}
A_{ij} = \frac{2 S_{(i)} - \rho 
}{\omegap} \Omegap_{(i)j}.
\end{equation}
Clearly, for spacetimes to be twisting, necessarily $\BS \neq \boldsymbol{0}$ and thus $\text{rank } \BS > 0$. At the same time, there has to be at least one non-vanishing off-diagonal component of $\Bomegap$. Hence from \eqref{14} we have that $\text{rank } \BS < d-2$. Therefore, $0 < \text{rank } \BS < d-2$ and the shear scalar $\sigma^2$ is necessarily non-vanishing, 
which proves the assertion for $d>4$. Using \eqref{16}, one immediately obtains the result also for $d=4$, i.e. $S_2$ and  $S_3$ have to be distinct for $A_{23}$
to be non-vanishing. 
\end{proof}
From \eqref{16}, we see that, in the frame in which $\BS$ takes the diagonal form, the number of non-vanishing components of $\BA$ is less than or equal to the number of non-vanishing off-diagonal components of $\Bomegap$. 
Thus, we conclude: 
\begin{remark}\label{cflattwist}
Twisting, conformally flat spacetimes 
with the Ricci tensor of the form \eqref{RicciIntro} do not exist.  
\end{remark}
\noindent 
Combining remarks \ref{cflat} and \ref{cflattwist}, we obtain that the following statement holds in arbitrary dimension.
\begin{corol}\label{c-flat}
A conformally flat 
spacetime with the Ricci tensor of the form \eqref{RicciIntro} is Kundt.
\end{corol}

In the following proposition, we show an allowed canonical form of the optical matrix for a twisting multiple WAND. 

\begin{propos}\label{twistingmatrix}
Let $\Bell$ be a multiple WAND of a twisting spacetime of  Weyl type N 
with the Ricci tensor of the form \eqref{RicciIntro}. 
Then the canonical form of the corresponding  optical matrix is 
\begin{equation} \label{generalform}
\boldsymbol{\rho} = s
\left(
\begin{array}{cc|c}
1 & a & \multirow{2}{*}{$\boldsymbol{0}$} \\
-a & b    \\ \hline
 \multicolumn{2}{c|}{\boldsymbol{0}} & \boldsymbol{0} \\
\end{array}
\right),
\end{equation}
where $b\not= 1$ (or otherwise the spacetime is  
Einstein)\footnote{
As in the non-twisting case, it follows from
\eqref{S3} or \eqref{ab} that $b=1$ iff $\omegat=0=R_{11}$ and thus an 
Einstein
spacetime corresponds to $b=1$.}
and $a$ and $b$ satisfy
\begin{equation}
a= \frac{ \Omegap_{23}}{\Omegap_{2} +\omegat}, \qquad
b= \frac{\Omegap_2 -(d-3) \omegat}{\Omegap_{2} + \omegat}=
{1-\frac{\omegap}{\Omegap_{2} + \omegat}}.\label{ab}
\end{equation} 
For $s=0$, it is Kundt.

The Weyl tensor frame components are
\begin{equation} \label{generalformWeyl}
\Omega'_{ij} = 
\left(
\begin{array}{cc|c}
\Omega'_2 & \Omega'_{23} & \multirow{2}{*}{$\boldsymbol{0}$} \\
\Omega'_{23} & (d-4)\omegat -\Omega'_2   \\ \hline
 \multicolumn{2}{c|}{\boldsymbol{0}} & -\omegat \BI \\ 
\end{array}
\right).
\end{equation}
\end{propos}

{Note that $\boldsymbol{\rho}$ \eqref{generalform} is again a normal matrix,
which does not satisfy the optical constraint.} Note also that relations 
\eqref{generalform}--\eqref{generalformWeyl} hold in the non-twisting case as well, where
we just substitute $a=0=\Omega'_{23}$.

\begin{proof}
We prove again only the case $d>4$. Let $\{ \hat \Be_{(a)}\}$ be the null frame, in which $\BS$ takes a diagonal form and let us  employ  the notation used in the previous sections again. Based on the above discussion, there are two possible scenarios: 
\begin{enumerate}[label=(\roman*)]
\item\label{jedna} There is exactly one non-vanishing $\Omegap_{ij}$ for some $i<j$;
\item \label{dva} There are at least two non-vanishing $\Omegap_{ij}$ for some $i<j$.
\end{enumerate} 
Let us start with the case \ref{dva}. Equation \eqref{14} implies that there is at most one non-vanishing diagonal component $S_i$ of $\BS$ and the only non-vanishing off-diagonal components of $\Bomegap$ can be those that appear in the same row or column as $S_i$. Otherwise, $\BS = \boldsymbol 0$ and thus 
$\BA = \boldsymbol 0$, which is a contradiction. Therefore, let us assume that $S_i$ is non-vanishing. Let us also rearrange $\{ \Bm_{(i)} \}$ so that $\BS = \text{diag}(S_2,0,\dots,0)$. 
Since the rest of the diagonal of $\BS$ vanishes, we have that $S_2 = \rho$,  
and thus from \eqref{14} only the off-diagonal components
 $\Omega'_{2i}$, $i>2$ (and thanks to \eqref{16}, also $A_{2i}$) can be non-vanishing.
From \eqref{11}, we have $S_2(\Omegap_j + \omegat) = 0$ for all $j>2$. Hence $\Omegap_j = -\omegat$ for all $j>2$. From the condition 
$\Tr \Bomegap = 0$, 
we conclude that $\Omegap_2 = (d-3) \omegat$. 
Using these relations for the components $\{ \Omegap_i\}$, from \eqref{2.50p},  
we obtain
\begin{equation}
\tho \Omegap_2 = -(d-3) \omegat \rho, 
\end{equation}
while for $i>2$ one has
\begin{equation}
\tho \Omegap_i =  \omegat \rho.  
\end{equation}
Considering \eqref{B1} for the choice of indices $j=i$ thus gives 
\begin{equation}\label{redukce}
\Omegap_{(i)k} \rho_{k(i)}  = 0.
\end{equation}
Since $\rho_{ki} = 0$ for all $i,k>2$, equation \eqref{redukce} reduces to 
\begin{equation}
\Omegap_{(i)2} A_{2(i)} = 0
\end{equation}
for all $i > 2$. This implies that $\BA = \boldsymbol{0}$, which again contradicts our assumption that the spacetime is twisting. Thus scenario \ref{dva} is excluded. 

Now, let us discuss case \ref{jedna}. According to \eqref{14}, the only possible non-vanishing components $S_i$ of $\BS$ are those which appear in the same row or column as the only non-vanishing  {off-diagonal} 
components $\Omegap_{ij}$ and $\Omegap_{ji}$ of $\Bomegap$. Again, we rearrange the spacelike vectors $\{ \Bm_{(i)} \}$ of the null frame such that 
$\BS = \text{diag}(S_2,S_3,0,\dots,0)$ and the non-vanishing {off-diagonal} components of $\Bomegap$ are $\Omegap_{23}$ and $\Omegap_{32}$.
We have that $S_{2} + S_3  = \rho $  
and, similar to the discussion of scenario \ref{dva}, from \eqref{11} 
we have $\Omegap_j = - \omegat$ for all $j > 3$. There are exactly $d-4$ such components $\Omegap_j$, thus the tracelessness of 
$\Bomegap$ implies that
$\Omegap_2 + \Omegap_3 = (d-4)\omegat $. 
Using these relations, we conclude that neither of $S_2$, $S_3$ vanishes and $S_2 \neq S_3$, otherwise 
\eqref{B1} and \eqref{7}, respectively, again implies  that $\Brho$ is twist-free.

The relations \eqref{ab} for parameters $a$ and $b$ can be derived
using 
\eqref{11}  
 or \eqref{S3}  
 and \eqref{16}.\footnote{Note that thanks to \eqref{Si}, the denominator $\Omegap_{2} + \omegat$ is always non-vanishing for $S_2\not=0$.}
\end{proof}

Putting propositions \ref{twistfree} and \ref{twistingmatrix} together, we observe that in both cases (twisting and the non-twisting), the optical matrix $\Brho$ 
takes the form  \eqref{generalform}, 
with $b\not= 1$. In the case of twisting $\Bell$, all functions  $s$, $a$ and $b$ must be non-vanishing, while in the non-twisting case ($a=0$),  $b$  or  $s$ may vanish.
Observe that for type N and generic type III {\it Einstein} spacetimes, as studied in \cite{bianchiclanek}, the optical matrix takes the form \eqref{generalform} with $b=1$.


\section{The  Ricci and Bianchi equations
and $r$-dependence }
\label{alleq}

Let us present all Ricci and Bianchi equations and commutators 
\cite{ghpclanek,ricciclanek}
for spacetimes of aligned  Weyl type N 
($\Bomegap \not=0$) 
with the Ricci tensor of the form \eqref{RicciIntro}  
($\omegap=\eta$, $\phi=\lambda$, 
$\phi_{ij}=\lambda\delta_{ij}$). Note that the multiple WAND $\bl$ is geodetic ($\kappa_i=0$). 
For simplicity, we  choose a parallelly propagated frame ($\tau'_i=0$, $\M{i}_{j0} =0$).

\subsection{Derivatives and commutators}

In a parallelly propagated frame, the \emph{GHP derivative operators} $\tho$, $\tho'$, $\eth_i$ \eqref{thorn}--\eqref{eth} act on a GHP scalar $\Tb$ of boost weight $b$ and spin $s$ as  
  \begin{eqnarray}
    \tho T_{i_1 i_2...i_s} &\equiv & D T_{i_1 i_2...i_s}, 
		\\
    \tho' T_{i_1 i_2...i_s} &\equiv & \Delta T_{i_1 i_2...i_s} - b L_{11} T_{i_1 i_2...i_s} 
                                     + \sum_{r=1}^s \M{k}_{i_r 1} T_{i_1...i_{r-1} k i_{r+1}...i_s},\\
    \eth_i T_{j_1 j_2...j_s} &\equiv & \delta_i T_{j_1 j_2...j_s} - b L_{1i} T_{j_1 j_2...j_s} 
                                     + \sum_{r=1}^s \M{k}_{j_r i} T_{j_1...j_{r-1} k j_{r+1}...j_s},
  \end{eqnarray}
where the derivatives along the frame vectors are defined in \eqref{deriv}.

The commutators (C1)--(C3) and (C2') in \cite{ghpclanek} read:
\begin{eqnarray}
[\tho, \tho']T_{i_1...i_s} 
         &=& \left( -\tau_j  \eth_j  
                   - b  \frac{\lambda}{d-1} 
             \right)T_{i_1...i_s} , \label{comm:thotho}\\ [3mm]
[\tho, \eth_i]T_{k_1...k_s}
         &=& -\rho_{ji}\eth_j T_{k_1...k_s} ,
            \label{comm:thoeth}\\[3mm]
[\eth_i,\eth_j]T_{k_1...k_s}
         &=& 2\left(\rho_{[ij]} \tho' + \rho'_{[ij]} \tho 
                   + b \rho_{l[i|} \rho'_{l|j]} \right) T_{k_1...k_s}\nonumber\\
         && + 2\sum_{r=1}^s \Big[\rho_{k_r [i|} \rho'_{l|j]} + \rho'_{k_r [i|} \rho_{l|j]} 
				+\frac{\lambda\delta_{[i|k_r}\delta_{|j]l}}{(d-1)} 
				\Big] 
				T_{k_1...l...k_s},
            \label{comm:etheth}\\[3mm]
						[\tho', \eth_i]T_{k_1...k_s}
         &=& \left[-(\kappa'_i \tho + \tau_i\tho' +\rho'_{ji}\eth_j)
             + b\left(-\tau_j\rho'_{ji} + \kappa'_j\rho_{ji} 
             \right) \right]T_{k_1...k_s} \nonumber\\
         & &+  \sum_{r=1}^s \Big[ \kappa'_{k_r}\rho_{li} - \rho'_{k_r i}\tau_l
            + \tau_{k_r} \rho'_{li} - \rho_{k_r i} \kappa'_l
            \Big] T_{k_1...l...k_s}.
            \label{comm:thoethprime}
\end{eqnarray}


\subsection{Bianchi equations}

The Bianchi equations (B2'), (B3'), (B4')=(B3')$_{i[jkl]}$ and (B1') in \cite{ghpclanek} 
 using (2.43') read (the remaining Bianchi equations are satisfied identically):
\begin{eqnarray}
 \tho (\Omega'_{ij} +\omegat \delta_{ij})
                 &=&  - (\Omega'_{ik}+\omegat \delta_{ik}) \rho_{kj},
                      \label{Bianchi-2}\\[3mm]
  0            &=&  - (\Omega'_{i[k|}+\omegat \delta_{i[k|}) \rho_{j|l]} 
	+ (\Omega'_{j[k|} +\omegat \delta_{j[k|})\rho_{i|l]},
                     \label{Bianchi-3}\\[3mm]
  0                 &=&  - (\Omega'_{i[j}+\omegat \delta_{i[j}) \rho_{kl]},\label{Bianchi-4}\\
  -  \eth_{[j}(\Omega'_{k]i} +\omegat \delta_{k]i})
                &=&  (\Omega'_{i[j}+\omegat \delta_{i[j}) \tau_{k]},\label{Bianchi-1}
\end{eqnarray}

The contracted Bianchi identity
$  \nabla^a R_{ab} = \pul\nabla_b R$ ((2.50') and (2.51)
in \cite{ghpclanek})  reads
\begin{eqnarray}
	\tho\omega'  &=& -\rho\omega', \label{BianchiMat1p}\\
  0    &=& 
	\kappa_i\omega'. \label{BianchiMat2}
\end{eqnarray}

\subsection{Ricci equations}

The Ricci equations (i.e. projections of
$v_{a;bc}-v_{a;cb}={R}_{sabc}v^s$  into the frame \eqref{frame}) 
for GHP scalars (equations (NP1)--(NP4) and (NP1')--(NP4')
in \cite{ghpclanek}) read
\begin{eqnarray}  
  \tho \rho_{ij} 
	&=& - \rho_{ik} \rho_{kj}, 
\label{fullsachs}\\
 \tho \tau_i  
	&=& - \rho_{ij}\tau_j, \label{R:thotau}\\[3mm]
    \eth_{[j|} \rho_{i|k]}     &=& \tau_i \rho_{[jk]}, 
																\label{R:ethrho} \\
 \tho' \rho_{ij} - \eth_j \tau_i &=& - \tau_i \tau_j
                                      - \rho_{ik}\rho'_{kj}
																			-\frac{\lambda}{d-1} \delta_{ij}, 	
																			\\
	 \tho \rho'_{ij} 
	&=&      - \rho'_{ik}\rho_{kj} 
																	-\frac{\lambda}{d-1} \delta_{ij},\label{R:thorhop}\\
																		- \tho \kappa'_i 
	&=& \rho'_{ij} \tau_j, 			\label{R:thotaup}\\[3mm]
   \eth_{[j|} \rho'_{i|k]}     &=&   \kappa'_i \rho_{[jk]},
                               	\label{R:ethrhop} \\
 \tho' \rho'_{ij} 
	- \eth_j \kappa'_i 
	&=& - \rho'_{ik} \rho'_{kj} 
	-\kappa'_i \tau_j  - \Omega'_{ij} - \omegat \delta_{ij}.\label{fullsachsp}
\end{eqnarray}

From (\ref{comm:thotho})--(\ref{comm:etheth}), we obtain the Ricci equations for ``non-GHP'' scalars
(i.e. NP scalars that do not transform covariantly as GHP scalars), see (11b), (11n),
(11m), (11a), (11d), (11p), (11c) and (11o) in \cite{ricciclanek}:
\bea
DL_{1i}  &=& -L_{1j} \rho_{ji}, \label{eqL1i}\\
 D \M{i}_{{j}{k}} &=&  - \M{i}_{{j}{l}} \rho_{lk}  ,\label{eqMijk}   \\
	D \M{i}_{{j}{1}} & =& - \M{i}_{{j}{k}} \tau_{k}  , \label{eqMij1} \\
 DL_{11} & =& - L_{1i} \tau_{i} 
  + \textstyle{\frac{\lambda}{d-1}}  ,  \label{eqL11}\\
  \delta_{[j|} L_{1|i]}& =& - L_{11} \rho_{[ij]}  
 -  L_{1k} \M{k}_{{[i}{j]}} 
- \rho_{k[j|} \rho'_{k|i]}  ,  \label{eqdelL1i}\\
   \delta_{[k|} \M{i}_{{j|}{l]}} 
&=&  \rho'_{i[l|} \rho_{j|k]}  
+ \rho_{i[l|} \rho'_{j|k]} +\rho_{[kl]} \M{i}_{{j}{1}} 
+   \M{i}_{{p}{[k|}} \M{p}_{{j|}{l]}} + \M{i}_{{j}{p}}  \M{p}_{{[k}{l]}} 
-\textstyle{\frac{\lambda}{d-1}}\delta_{i[k}  \delta_{l]j},\label{eqdelMijk}\\
 \bigtriangleup L_{1i} - \delta_i L_{11} &=&  L_{11}(L_{1i}- \tau_{i})
 -  \tau_{j} \rho'_{ji}+\rho_{ji} \kappa'_{j} 
-L_{1j} (\rho'_{ji}+\M{j}_{{i}{1}}) , \label{eqnabL1idelL11}\\
  \bigtriangleup \M{i}_{{j}{k}} - \delta_{k} \M{i}_{{j}{1}} 
&=& 
2 \kappa'_{[j} \rho_{i]k}    
+2\tau_{[j} \rho'_{i]k} 
+ \M{i}_{{j}{1}} (L_{1k}-\tau_{k})
+ 2  \M{i}_{{l}{[1|} }\M{l}_{{j|}{k]}} - \M{i}_{{j}{l}} (\rho'_{lk} + \M{l}_{{k}{1}}  ) .
\label{eqnabMijkdelMij1}
\eea

\subsection{The $r$-dependence}

To integrate some of these equations, let
 $r$ denote an affine parameter along null geodesics generated by the multiple WAND $\Bell$ and let us work in a parallelly propagated 
frame.\footnote{The form (\ref{generalform}) of the optical matrix is compatible with the
parallel transport \cite{OrtPraPra10}.} 
Then the operator $\tho$  \eqref{thorn} reduces to the derivative along
the integral curves of $\bl$, i.e. $\tho f=Df$ for a function $f$. 
The Sachs equation (\ref{fullsachs})
for the optical matrix of the form (\ref{generalform}) can be integrated
(see \cite{OrtPraPra10}, where the $r$-dependence 
of the optical matrix of a block form has been derived): 
\bea
 & &  
s_{p}=\R\left[s_{0p}+r(s_{02}s_{03}+a_0^2)\right], 
  \qquad p=2,3, \label{s_aligned} \\
& & A_{23}= \R a_0 ,\ \ \mbox{where}\ \ 
\R=\frac{1}{1+r\big(s_{02}+s_{03}\big)+r^2(s_{02}s_{03}+a_0^2)} ,
\label{A_aligned}
\eea
 where $s_{0p}$ and $a_0$ do not depend on $r$. Note that $s_2=s$, $s_3=sb$
and $A_{23}=sa$, where $s$, $b$ and $a$ are defined in \eqref{ab}.

Integrating \eqref{2.50p}$=$\eqref{BianchiMat1p} 
and (\ref{B1})$=$(\ref{Bianchi-2}), and using (\ref{s_aligned}, \ref{A_aligned}), we may further obtain the $r$-dependence of the non-trivial null frame components of the Ricci 
\eqref{RicciIntro} and Weyl tensors \eqref{generalformWeyl},
respectively:
\begin{align}
&{\tilde \omega}= {\tilde \omega}_0 \R,\\
\begin{split}
&\Omega'_{22}=({o_{02}}r+o_{01}) \R ,\\
&\Omega'_{23}=-\frac{\R}{a_0^2} \Big\{ r\left[({\tilde \omega}_0+o_{01})(s_{02}s_{03}+a_0^2)
-o_{02} s_{02}\right]+({\tilde \omega}_0+o_{01})s_{03}-o_{02} \Big\} ,\\
\begin{split}
\Omega'_{33} =-\frac{\R}{a_0^2}  &\Big\{  
r \left[ ( {\tilde \omega}_0+o_{01})(s_{02}s_{03}+a_0^2)(s_{02}-s_{03})+o_{02}(a_0^2+s_{02}s_{03}- s_{02}^2)\right]   \\
&\:+ ({\tilde \omega}_0+o_{01})(a_0^2+s_{02}s_{03}- s_{03}^2)+{\tilde \omega}_0 a_0^2+o_{02}(s_{03}-s_{02}) 
\Big\}, \\
\end{split}\\
&\Omega'_{(w)(w)}= - {\tilde \omega}_0 \R,
\end{split}
\end{align}
where ${\tilde \omega}_0 $, $o_{01}$ and $o_{02}$ are independent of $r$ and 
$w=4,\dots, d-1$.

 As in the Einstein spacetime case,  curvature singularities will appear (see equation (29) in \cite{OrtPraPra10} and the discussion below).

One can further integrate (\ref{R:thotau}) and (\ref{R:thorhop}): 
\bea
\tau_2&=&(\tau_{2}^{(1)}r+\tau_{2}^{(0)})\R,\\
\tau_3&=&-\frac{1}{a_0}\left\{ 
r\left[\tau_{2}^{(1)}s_{02}-\tau_{2}^{(0)}(s_{02}s_{03}+a_0^2)\right]
+\tau_{2}^{(1)}-\tau_{2}^{(0)} s_{03}\right\}\R,\\
\tau_w&=&\tau_w^{(0)},\\
\rho'_{22}&=&({\rho'_{22}}\ ^{\!\!\! (1)} r+{\rho'_{22}}\ ^{\!\!\! (0)})\R
-\frac{\lambda r}{2(d-1)}\left[1+(r s_{03}+1)\R\right],\\
\rho'_{23}&=&
\frac{1}{a_0}\left\{r\left[{\rho'_{22}}\ ^{\!\!\! (1)}s_{02}
-{\rho'_{22}}\ ^{\!\!\! (0)}(s_{02}s_{03}+a_0^2)\right]
+{\rho'_{22}}\ ^{\!\!\! (1)}-{\rho'_{22}}\ ^{\!\!\! (0)}s_{03}
+\frac{\lambda a_0^2 r^2}{2(d-1)}\right\}\R,\\
\rho'_{33}&=&({\rho'_{33}}\ ^{\!\!\! (1)} r+{\rho'_{33}}\ ^{\!\!\! (0)})\R
-\frac{\lambda r}{2(d-1)}\left[1+(r s_{02}+1)\R\right],\\
\rho'_{32}&=&
-\frac{1}{a_0}\left\{r\left[{\rho'_{33}}\ ^{\!\!\! (1)}s_{03}
-{\rho'_{33}}\ ^{\!\!\! (0)}(s_{02}s_{03}+a_0^2)\right]
+{\rho'_{33}}\ ^{\!\!\! (1)}-{\rho'_{33}}\ ^{\!\!\! (0)}s_{02}
+\frac{\lambda a_0^2 r^2}{2(d-1)}\right\}\R,\\
\rho'_{v2}&=&({\rho'_{v2}}\ ^{\!\!\! (1)} r+{\rho'_{v2}}\ ^{\!\!\! (0)})\R,
\\
\rho'_{v3}&=&\frac{1}{a_0}\left\{r\left[{\rho'_{v2}}\ ^{\!\!\! (1)}s_{02}
-{\rho'_{v2}}\ ^{\!\!\! (0)}(s_{02}s_{03}+a_0^2)\right]
+{\rho'_{v2}}\ ^{\!\!\! (1)}-{\rho'_{v2}}\ ^{\!\!\! (0)}s_{03}\right\}\R ,\\
\rho'_{iw}&=&-\frac{\lambda \delta_{iw}}{d-1}r+{\rho'_{iw}}\ ^{\!\!\! (0)},
\eea
where functions with the superscript $(0)$ or $(1)$ do not depend on $r$ and $v,w= 4, \dots,
d-1$.

Equations (\ref{eqL1i}) and (\ref{eqMijk}) can be also integrated:
\bea
L_{12}&=&(l_{12}^{(1)}r+l_{12}^{(0)})\R,\\
L_{13}&=&\frac{1}{a_0}\left\{ 
r\left[l_{12}^{(1)}s_{02}-l_{12}^{(0)}(s_{02}s_{03}+a_0^2)\right]
+l_{12}^{(1)}-l_{12}^{(0)} s_{03}\right\}\R,\\
L_{1w}&=&l_{1w}^{(0)},\\
\M{i}_{j2}&=&(\m{i}_{j2}\ ^{\!\!\! (1)}r+\m{i}_{j2}\ ^{\!\!\! (0)})\R,\\
\M{i}_{j3}&=&\frac{1}{a_0}\left\{ 
r\left[\m{i}_{j2}\ ^{\!\!\! (1)}s_{02}-\m{i}_{j2}\ ^{\!\!\! (0)}(s_{02}s_{03}+a_0^2)\right]
+\m{i}_{j2}\ ^{\!\!\! (1)}-\m{i}_{j2}\ ^{\!\!\! (0)} s_{03}\right\}\R,\\
\M{i}_{jw}&=&\m{i}_{jw}\ ^{\!\!\! (0)}.
\eea
																
The $r$-dependence of $\kappa'_i$, $\M{i}_{{j}{1}}$  and $L_{11}$ follows from  equations 
\eqref{R:thotaup}, (\ref{eqMij1}) and (\ref{eqL11}), respectively, and is given by 
\bea
\kappa'_i&=&-\int\rho'_{ij} \tau_j{\rm d} r , \\
	 \M{i}_{{j}{1}} & =& - \int \M{i}_{{j}{k}} \tau_{k} {\rm d} r , \label{Mij1} \\
 L_{11} & =& - \int L_{1i} \tau_{i} {\rm d} r
  + \textstyle{\frac{\lambda r}{d-1}}  .  \label{L11}
\eea

To investigate the consequences of the remaining equations, one would need to introduce
the remaining coordinates and the rest of the frame, i.e. 
the frame vectors $\bn$, $\Bm_{(i)}$, similarly as was done in \cite{5DHarvey-notwist}, \cite{5DHarvey-Kerr} and \cite{5DHarvey-twist}, and this is 
is not in the scope of this paper.
\section{
{Examples: Direct products of aligned Ricci type N spacetimes  and 
$\mathbb{E}^{n}$}}\label{priklad}

Higher-dimensional examples of aligned Weyl type N, Ricci type N spacetimes 
can be easily constructed as 
direct products of four-dimensional Weyl type N, pure radiation spacetimes with  flat 
Euclidean spaces.

For a direct product $M={\tilde M}\times N$ of two arbitrary manifolds (Lorentzian 
${\tilde M}$ and Riemannian $N$),    the metric, as well as the Christoffel symbols and the Ricci and Riemann tensors, are decomposable.
Then, obviously, for any tensor $\tilde{\boldsymbol{T}}$ from ${\tilde M}$ lifted to $M$, the covariant derivative in $M$ reduces to the covariant derivative in 
${\tilde M}$, i.e.
$\nabla {\boldsymbol{T}} = \tilde{\nabla}\tilde{\boldsymbol{T}}$.

Applying this result to a direct product of an aligned algebraically special  Ricci type N spacetime of dimension $m$ with the Euclidean space $\mathbb{E}^n$ for some $n \in \mathbb{N}$, 
we obtain that $\nabla \Bell$ reduces to $\tilde \nabla \tilde \Bell$ and the form of the optical matrix \eqref{directmatrix} follows (where frames from
${\tilde M}$ and $N$  lifted to $M$ are used).
Optical scalars thus read\footnote{Note that it is a special form of (17) in \cite{typeD} for $f=$const.}
\begin{equation}\label{opticalscalarsdirect}
\theta = \frac{m-2}{d-2}\tilde \theta,\qquad 
\sigma^2 = \tilde \sigma^2 + \frac{(m-2)(d-m)}{d-2} \tilde \theta^2, \qquad
A^2 = \tilde A^2.
\end{equation}
Note that $\Bell$ is always shearing, unless the original spacetimes is Kundt. 
Moreover, $\Bell$ is expanding and twisting if and only if $\tilde \Bell$ is expanding and twisting, respectively, in the original spacetime. In addition, the resulting spacetime is Kundt if and only if the original spacetimes is Kundt.

The alignment type of the Weyl tensor
\begin{equation}\label{weylexpr}
C_{abcd} = R_{abcd} - \frac{2}{d-2}(g_{a[c} R_{d]b} - g_{b[c} R_{d]a}) + 
\frac{2}{(d-1)(d-2)}R g_{a[c}g_{d]b}
\end{equation} 
is determined by the maximal boost weight components of both Weyl tensors and both Ricci tensors (see (12)--(14) in \cite{typeD}). For non-flat $N$, the resulting Weyl type would  in general be II.
For flat Euclidean $\mathbb{E}^n$ and non-conformally flat ${\tilde M}$, the resulting Weyl type is the same as that of $\tilde{\boldsymbol{C}}$ and the Ricci tensor
remains as
\begin{equation}\label{directricci}
R_{ab} = \omegap \ell_a \ell_b,
\end{equation}
i.e. a direct product of Weyl type N, Ricci type N spacetime with a Euclidean space is
always of Weyl type N.

So, let us summarize:
\begin{propos}\label{directtype}
Consider a direct product of an aligned algebraically special Ricci type N spacetime and $\mathbb{E}^n$. 
The following statements hold:
\begin{enumerate}[label=(\roman*)]
\item The Weyl type does not change, unless the original spacetime is conformally flat. \\
If this is the case, the Weyl type changes from O to N.
The  mWAND $\Bell$ is the lift of the mWAND $\tilde \Bell$ in the original spacetime.
\item The Ricci tensor in the resulting spacetime is of the type N and it 
is aligned with the Weyl tensor. 
\item The optical matrix $\Brho$ corresponding to $\Bell$ admits the form 
\begin{equation}\label{directmatrix}
\boldsymbol{\rho} = 
\left(
\begin{array}{c|c}
\tilde\Brho & \boldsymbol{0} \\ \hline
\boldsymbol{0} & \boldsymbol{0}    \\ 
\end{array}
\right),
\end{equation}
where $\tilde \Brho$ is the optical matrix corresponding to $\tilde \Bell$ in the original spacetime.
\end{enumerate} 
\end{propos}

Let us conclude with references to various classes of  four-dimensional Weyl type  N, Ricci type N 
spacetimes in table \ref{4Dexamples}. 
\begin{table}[htb]
  \begin{center}
  \begin{tabular}{|l|l|} 
    \hline $\boldsymbol{\ \ \ \ \ \ \  \ \ \ \ \  \ \ \ \ \ \ \ \rho}$ 
		& $\ \ \ \ \ \ \ \ \ \ \ \ \ \ \ \ \ \ \  \  {\tilde M}$   \\ \hline 
    $\boldsymbol{\rho}=\boldsymbol{0}$, i.e. Kundt case &    
		 \cite[ (18.8)-(18.9)]{griffiths},
\cite{ORR} {\small (The solution found in \cite{genKundt} is not included.)}\\
 $A_{ij}=0$, i.e. non-twisting case& \cite[(26.11a,b)]{stephani} \\
		$A_{ij}\not=0$, i.e. twisting case & \cite[(64)-(67)]{Plebanski2}, 
		\cite{Plebanski}\\
\hline
  \end{tabular}
  \caption{Four-dimensional  Weyl type N and  Ricci type N or O spacetimes.}
  \label{4Dexamples}
   \end{center}
\end{table}

\section*{Acknowledgments}
The authors are thankful to M. Ortaggio, V. Pravda and both referees for useful comments on the draft.




\begin{thebibliography}{10}

\bibitem{Milsonetal05}
R.~Milson, A.~Coley, V.~Pravda, and A.~Pravdov{\'a}.
\newblock Alignment and algebraically special tensors in {Lorentzian} geometry.
\newblock {\em Int. J. Geom. Methods Mod. Phys.},
  2(01):41--61, 2005.

\bibitem{Coleyetal04}
A.~Coley, R.~Milson, V.~Pravda, and A.~Pravdov{\'a}.
\newblock Classification of the {Weyl} tensor in higher dimensions.
\newblock {\em Class. Quantum Grav.}, 21(7):L35, 2004.

\bibitem{review}
M.~Ortaggio, V.~Pravda, and A.~Pravdov{\'a}.
\newblock Algebraic classification of higher dimensional spacetimes based on
  null alignment.
\newblock {\em Class. Quantum Grav.}, 30(1):013001, 2013.


\bibitem{Ortaggioetal12}
M.~Ortaggio, V.~Pravda, A.~Pravdov{\'a}, and H.~S. Reall.
\newblock On a five-dimensional version of the {Goldberg}--{Sachs} theorem.
\newblock {\em Class. Quantum Grav.}, 29(20):205002, 2012.

\bibitem{5DHarvey-notwist}
H.~S. Reall, A.~A.~H. Graham, and C.~P. Turner.
\newblock On algebraically special vacuum spacetimes in five dimensions.
\newblock {\em Class. Quantum Grav.}, 30:055004, 2013.

\bibitem{5DHarvey-Kerr}
G.~B. de~Freitas, M.~Godazgar, and H.~S. Reall.
\newblock Uniqueness of the Kerr--de Sitter spacetime as an algebraically
  special solution in five dimensions.
\newblock {\em Commun. Math. Phys.}, 340:291-323, 2015.



\bibitem{5DHarvey-twist}
G.~B. de~Freitas, M.~Godazgar, and H.~S. Reall.
\newblock Twisting algebraically special solutions in five 
dimensions.
\newblock {\em {\tt arXiv:1511.02238 }}, 2015.





\bibitem{stephani}
H.~Stephani, D.~Kramer, M.~MacCallum, C.~Hoenselaers, and E.~Herlt.
\newblock {\em Exact solutions of {Einstein's} field equations}.
\newblock Cambridge University Press, 2003.

\bibitem{griffiths}
J.~B. Griffiths and J.~Podolsk{\'y}.
\newblock {\em Exact space-times in {Einstein's} general relativity}.
\newblock Cambridge University Press, 2009.

\bibitem{MalPra11}
T.~M{\'a}lek and V.~Pravda.
\newblock {Kerr}--{Schild} spacetimes with an {(A)dS} background.
\newblock {\em Class. Quantum Grav.}, 28(12):125011, 2011.

\bibitem{turek1}
M.~Gurses, T.~C. Sisman, and B.~Tekin.
\newblock AdS-plane wave and pp-wave solutions of generic gravity theories.
\newblock {\em Phys. Rev.  D}, 90:124005, 2014.

\bibitem{turek2}
M.~Gurses, S.~Hervik, T.~C. Sisman, and B.~Tekin.
\newblock Anti-de Sitter-wave solutions of higher derivative theories.
\newblock {\em Phys. Rev.  Letters}, 111:10111, 2013.

\bibitem{wils}
P.~Wils.
\newblock Homogeneous and conformally {Ricc}i flat pure radiation fields.
\newblock {\em Class. Quantum Grav.}, 6(9):1243, 1989.

\bibitem{bianchiclanek}
V.~Pravda, A.~Pravdov{\'a}, A.~Coley, and R.~Milson.
\newblock Bianchi identities in higher dimensions.
\newblock {\em Class. Quantum Grav.}, 21(12):2873, 2004.

\bibitem{bianchiclanekcor}
V.~Pravda, A.~Pravdov{\'a}, A.~Coley, and R.~Milson.
\newblock Corrigendum: Bianchi identities in higher dimensions.
\newblock {\em Class. Quantum Grav.}, 24:1691, 2007.

\bibitem{ghpclanek}
M.~Durkee, V.~Pravda, A.~Pravdov{\'a}, and H.~S. Reall.
\newblock Generalization of the {Geroch}--{Held}--{Penrose} formalism to higher
  dimensions.
\newblock {\em Class. Quantum Grav.}, 27(21):215010, 2010.

\bibitem{OrtPraPra09}
M.~Ortaggio, V.~Pravda, and A.~Pravdov{\'a}.
\newblock Higher dimensional {Kerr}--{Schild} spacetimes.
\newblock {\em Class. Quantum Grav.}, 26(2):025008, 2009.

\bibitem{OrtPraPra09b}
M.~Ortaggio, V.~Pravda, and A.~Pravdov{\'a}.
\newblock Asymptotically flat, algebraically special spacetimes in higher
  dimensions.
\newblock {\em Phys. Rev.  D}, 80(8):084041, 2009.

\bibitem{clanek11}
M.~Ortaggio, V.~Pravda, and A.~Pravdov{\'a}.
\newblock On the {Goldberg}--{Sachs} theorem in higher dimensions in the
  non-twisting case.
\newblock {\em Class. Quantum Grav.}, 30(7):075016, 2013.

\bibitem{OrtPraPra10}
M.~Ortaggio, V.~Pravda, and A.~Pravdov{\'a}.
\newblock Type {III} and {N} {Einstein} spacetimes in higher dimensions:
  general properties.
\newblock {\em Phys. Rev.  D}, 82(6):064043, 2010.

\bibitem{ricciclanek}
M.~Ortaggio, V.~Pravda, and A.~Pravdov{\'a}.
\newblock Ricci identities in higher dimensions.
\newblock {\em Class. Quantum Grav.}, 24(6):1657, 2007.


\bibitem{HDVSI}
A.~Coley, A.~Fuster, S.~Hervik, and N.~Pelavas.
\newblock Higher dimensional VSI spacetimes.
\newblock {\em Class. Quantum Grav.}, 23:7431, 2006.

\bibitem{typeD}
V.~Pravda, A.~Pravdov{\'a}, and M.~Ortaggio.
\newblock Type D Einstein spacetimes in higher dimensions.
\newblock {\em Class. Quantum Grav.}, 24:4407--4428, 2007.

\bibitem{ORR}
I.~Ozsv{\'a}th, I.~Robinson, and K.~R{\'o}zga.
\newblock Plane-fronted gravitational and electromagnetic waves in spaces with
  cosmological constant.
\newblock {\em J. Math.  Phys.}, 26(7):1755--1761, 1985.

\bibitem{genKundt}
J.~B. Griffiths, P.~Docherty, and J.~Podolsk{\'y}.
\newblock Generalised Kundt waves and their physical interpretation.
\newblock {\em Class. Quantum Grav.}, 21:207--222, 2004.


\bibitem{Plebanski2}
J.~F. Pleba{\'n}ski and K.~R{\'o}zga.
\newblock On certain twisting type-{N} solutions of {Einstein} equations with
  pure radiation energy-momentum tensor and nonvanishing shear.
\newblock {\em Class. Quantum Grav.}, 19(23):6025, 2002.

\bibitem{Plebanski}
J.~F. Pleba{\'n}ski, M.~Przanowski, and S.~Forma{\'n}ski.
\newblock Linear superposition of two type-{N} nonlinear gravitons.
\newblock {\em Phys. Lett. A}, 246(1):25--31, 1998.


\end{thebibliography}
\end{document}